\documentclass[11pt,letterpaper]{article}

\usepackage{amssymb}
\usepackage{amsmath}
\usepackage{subfigure}
\usepackage{cite}
\usepackage{graphicx}
\usepackage{graphics}
\usepackage{color}
\usepackage{xspace}
\usepackage{psfrag}
\usepackage{algorithmicx}
\usepackage{algorithm}
\usepackage{algpseudocode}
\usepackage{multirow}
\usepackage{array}
\usepackage{url}

\usepackage{amsmath,amsthm}
\usepackage{amssymb}
\usepackage{floatflt,setspace}
\newtheorem{theorem}{Theorem}
\newtheorem{definition}{Definition}

\newtheorem{lemma}{Lemma}

\title{Information-theoretic Bounds on Matrix Completion under Union of Subspaces Model}
\author{Vaneet Aggarwal and Shuchin Aeron\thanks{V. Aggarwal is with Purdue University, W. Lafayette, IN 47907, email: vaneet@purdue.edu. This work of V. Aggarwal was supported in part by Air Force Research Lab Visiting Faculty Research Program Award.
 S. Aeron is with Tufts University, Medford, MA 02155, email: shuchin@ece.tufts.edu. S. Aeron is supported in part by NSF grant 1319653.} }
\date{\vspace{-5ex}}

\begin{document}
\maketitle
\begin{abstract}
In this short note we extend some of the recent results on matrix completion under the assumption that the columns of the matrix can be grouped (clustered) into subspaces (not necessarily disjoint or independent). This model deviates from the typical assumption prevalent in the literature dealing with compression and recovery for big-data applications. The results have a direct bearing on the problem of subspace clustering under missing or incomplete information.

\end{abstract}
\section{Introduction}
Matrix  completion  refers  to  the  recovery  of  a  low-rank matrix   from  a   (small)  subset   of   its   entries  or   a   (small) number  of  linear  combinations  of  its  entries \cite{candes,gross,recht,Cai}. In essence, the methods are aimed at recovering the column/row subspaces from limited measurements. Even the sketching methods \cite{Drineas} aim to find the best column (or row) subspace of a matrix.

However, in many practical applications, the columns of the data matrix can belong to different low rank subspaces (or affine subspaces) \cite{Kanade,Kana,now1,vidal}. Motivated by this observation, in this paper we assume that the different columns in the data matrix of size $m\times n$ lie in one of the $K$ subspaces, where the dimension of these subspaces are $(r_1, \cdots, r_K)$. Now, suppose we have $k$ linear measurements of this matrix. The general question is how many linear measurements, satisfying certain properties are sufficient such that the data can be recovered from these linear measurements.


This problem has direct bearing on the problem of subspace clustering \cite{vidal,sujay1,bol2,wang,aarti} under missing or incomplete data. Subspace clustering with missing data has been studied in \cite{bal1,bal2}. Recently, the authors of \cite{bal3} considered the number of samples needed for reconstruction of data, where the number of partially observed data vectors per subspace is a rank-degree power of the dimension. In contrast to these results, in this paper we show information-theoretically, number of linear measurements greater than $Kr(m+n/K-r)$ suffice for reconstruction of data when the columns are assumed to come from a union of $K$ subspaces each of dimension $r$. Further, we note that rank-1 measurement matrices are enough over the whole data matrix.


Our main tool to obtain the sufficiency result relies on recent information-theoretic results on matrix completion in \cite{bol}, inspired by fundamental limits on analog source compression in \cite{yih}. In this paper, we specialize these results to the union of subspaces model.

The rest of the paper is organized as follows. Section II gives the system model, and Section III gives the sufficient number of linear measurements. Section IV concludes this paper.

\section{Model and Preliminaries}
For notations, let Roman  letters $A,B,\cdots $ designate  deterministic matrices  and $a,b,\cdots$ stands  for  deterministic  vectors.  Bold-face  letters ${\bf A}, {\bf B},\cdots$ and ${\bf a}, {\bf b}, \cdots$ denote random matrices  and vectors, respectively. Let ${\cal M}^{m\times n}_r$ and ${\cal N}^{m\times n}_r$ denote  the  set  of  matrices $A\in {\mathbb R}^{m\times n}$ with $\text{rank}(A)\le r$ and $\text{rank}(A) =r$, respectively. For  a  random  matrix ${\bf X}\in R^{m\times n}$ of  arbitrary distribution $\mu_{\bf X}$, an $(m\times n, k)$ code consists of  linear  measurements ($<A_1,\cdot>, \cdots, <A_k,\cdot>$, where $<A,B>$ is the trace inner product between $A$ and $B$) $T:R^{m\times n} \to R^k$, and  a measurable decoder $g:R^k\to R^{m\times n}$. For given measurement matrices $A_i$, we say that a decoder $g$ achieves error probability $\epsilon$ if $\Pr (g((<A_1,{\bf X}>, \cdots, <A_k,{\bf X}>)^T)\ne {\bf X})\le \epsilon$. For $\epsilon>0$,  we  call  a  nonempty  bounded  set ${\cal S} \in {\mathbb R}^{m\times n}$ an $\epsilon$-support set of the random matrix ${\bf X}\in R^{m\times n}$ if $\Pr[{\bf X} \in  {\cal S}]\ge 1 - \epsilon$.

We next define Minkowski dimension.
\begin{definition}[\cite{bol}]
Let ${\cal S}$ be a nonempty bounded set in ${\mathbb R}^{m\times n}$. The lower Minkowski dimension of ${\cal S}$ is defined as
\begin{equation}\label{dim1}
  \underline{\text{dim}}({\cal S}) = \liminf_{\rho \to 0} \frac{\log N_{\cal S}(\rho)}{-\log \rho},
\end{equation}
and the upper Minkowski dimension of ${\cal S}$ is defined as
\begin{equation}\label{dim2}
  \overline{\text{dim}}({\cal S}) = \limsup_{\rho \to 0} \frac{\log N_{\cal S}(\rho)}{-\log \rho},
\end{equation}
where $N_{\cal S}$ denotes the covering number of ${\cal S}$ given by
\begin{equation}\label{cn}
  N_{\cal S}(\rho) = \min\{k \in {\mathbb N}: {\cal S}\subseteq\cup_{i\in \{1, \cdots, k\} } {\cal B}_{m\times n} (M_i,\rho), M_i \in {\mathbb R}^{m\times n} \},
\end{equation}
and  ${\cal B}_{k}(\mu,s)$ denotes  the  open  ball  of  radius $s$ centered  at $\mu \in {\mathbb R}^k$.
\end{definition}

We next give a bound of number of measurements needed to decode matrix from limited measurements.
\begin{lemma}[\cite{bol}]\label{dim_needed}
Let ${\cal S}  \subseteq {\mathbb R}^{m\times n}$ be  an $\epsilon$-support  set  of ${\bf X}\in {\mathbb R}^{m\times n}$. Then, for Lebesgue a.a. measurement matrices $A_i, i=1\cdots, k$,  there  exists  a  decoder  achieving  error  probability $\epsilon$,provided that $k >\underline{\text{dim}}({\cal S})$
\end{lemma}

We will now describe the union of subspace model that is considered in this paper.
\begin{definition}
The union of subspace set ${\cal US}^{m\times n}_{K, (r_1, \cdots, r_K)}$ is the of matrices $X$ for which its columns can be divided among $K$ groups to get $X_1, \cdots X_K$, where each column of $X$ is in exactly one $X_i$, and $X_i \in {\cal M}^{m\times n_i}_{r_i}$, where $n_i\ge 0, \sum_{i=1}^K n_i = n$.
\end{definition}


\section{Main Results}
\begin{theorem}\label{thm:main}
  Let ${\cal S}\subseteq {\cal US}^{m\times n}_{K, (r_1, \cdots, r_K)}$ be a non-empty bounded set. Then,
  \begin{equation}\label{main}
    \overline{\text{dim}}({\cal S}) \le m \sum_i r_i + n \max_i r_i - \sum_i r_i^2
  \end{equation}
\end{theorem}
\begin{proof}
  We can represent $X \in {\cal US}^{m\times n}_{K, (r_1, \cdots, r_K)}$ with a set of columns $C_i, i = 1, \cdots, K$ for the $K$ subspaces with $|C_i|=n_i$, and $X({C_i})\in {\cal M}^{m\times n_i}_{r_i}$ represents the $X$ in those columns. This, ${\cal US}^{m\times n}_{K, (r_1, \cdots, r_K)}$ is equivalent to
  \[\cup_{C_1, \cdots, C_K} \times_{i=1}^K {\cal M}^{m\times |C_i|}_{r_i}\]
  where $\times$ refers to the Cartesian product. Therefore the manifold of union of subspaces is a \emph{product manifold}. Since upper Minkowski dimension for ${\cal M}^{m\times |C_i|}_{r_i}$  is at most $r_i(m+|C_i|-r_i)$ \cite{bol}, the upper Minkowski dimension for $\times_{i=1}^K {\cal M}^{m\times |C_i|}_{r_i}$ is at most $\sum_i(r_i(m+|C_i|-r_i))$.

Further, since upper Minkowski dimension of union is the max of the Minkowski dimension [Section 3.2, \cite{fra}], we have

\begin{eqnarray}
    \overline{\text{dim}}({\cal S}) &\le& \max_{C_1, \cdots, C_K}\sum_i(r_i(m+|C_i|-r_i))\\
    &=& \sum_ir_im   +  \max_{C_1, \cdots, C_K}\sum_i r_i |C_i|-\sum_i r_i^2\\
    &\le& m \sum_i r_i + \max_{C_1, \cdots, C_K}\sum_i (\max_i r_i)|C_i| - \sum_i r_i^2\\
    &=& m \sum_i r_i +  n(\max_i r_i) - \sum_i r_i^2
  \end{eqnarray}
\end{proof}

We note the following points.

\begin{enumerate}
\item  Since $\underline{dim}\le \overline{dim}$, from Theorem \ref{thm:main} and Lemma \ref{dim_needed}, we see that there  exists  a  decoder  that  achieves  error  probability $\epsilon$ for  $k$ Lebesgue  a.a.  measurement  matrices, for $k>m \sum_i r_i + n \max_i r_i - \sum_i r_i^2$.

\item {In the special case when the subspaces are independent and when $r_i = r$ for all $i$, we have the sufficient number of linear measurements for the UOS model as $Kr(m-r)+ n r+1$. Under a single subspace model, the number of measurements for a matrix with rank $Kr$ (which will be total dimension of the space spanned by the columns under the independence assumption and assuming $n \geq Kr$) is $ (m + n - Kr)Kr+1$. The difference in the number of linear measurements is $(n-Kr)(K-1)r$. When $n = Kr$ this tells us that there is no advantage in using the UOS model as compared to a single subspace model.}

 \item Note that there is no additional overhead in number of measurements for the knowledge of $K$ sets of columns that make each subspace since the number of measurements are equivalent to measuring each of $K$ subspaces knowing which columns make each subspace. We also note that with exact completion with these measurements, subspace clustering can be performed \cite{aarti} to also get different subspace clusters.

\item We further note from Theorem 2 of \cite{bol} that rank one measurement matrices are sufficient, rather than general linear measurement matrices. Rank one measurement matrices are attractive as they require less storage space than general measurement matrices and   can   also   be   applied  faster.

\item We note that rank one measurements are used over the whole data rather than performing subspace clustering with limited measurements, followed by performing measurements in each subspace. 
\end{enumerate}


\section{Conclusion}
This paper finds the number of linear measurements that are sufficient to estimate a matrix that is formed by a union of subspaces. The savings of measurements with the additional structure of union of subspace model depend on the product of dimension, number of subspaces, and the rank of subspace.

\bibliographystyle{IEEETran}

\small
\begin{thebibliography}{10}
\bibitem{candes}E. J. Candes and Y. Plan, ``Tight oracle inequalities for
low-rank matrix
recovery from a minimal number of noisy random measurements,''
{\em IEEE
Trans. Inf. Theory}, vol. 4, no. 57, pp. 2342--2359, Apr. 2011.

\bibitem{gross} D. Gross, ``Recovering low-rank matrices from few coefficients in any
basis,'' {\em IEEE Trans. Inf. Theory}, vol. 57, no. 3, pp. 1548--1566, Mar.
2011.

\bibitem{recht}
B. Recht, ``A simpler approach to matrix completion,'' {\em J. Mach. Learn.
Res.}, vol. 12, pp. 341--3430, 2011.

\bibitem{Cai} T. T. Cai and A. Zhang, ``ROP: Matrix recovery via rank-one projections,'' {\em Ann. Stat.}, vol. 43, no. 1, pp. 102–138, 2015.

\bibitem{Kanade}J. P. Costeira and T. Kanade, ``A multi-
body factorization method for independently moving
objects,'' {\em International Journal of Computer Vision}, vol.
29, 1998.


\bibitem{Kana} K. Kanatani, ``Motion Segmentation by Subspace Sep-
aration and Model Selection,'' in Proc. {\em  IEEE Interna-
tional Conference on Computer Vision,} 2001, vol. 2, pp. 586--591.

\bibitem{now1}
B. Eriksson, P. Barford, J. Sommers, and R.
Nowak, ``DomainImpute: Inferring Unseen Components in the Internet,'' in Proc. {\em IEEE INFOCOM Mini-Conference}, April 2011, pp. 171--175.


\bibitem{Drineas} Christos Boutsidis, Petros Drineas, and Malik Magdon-Ismail, `` Near-Optimal Column-Based Matrix Reconstruction'', SIAM Journal on Computing 2014 43:2, 687-717.

\bibitem{vidal}
R. Vidal, ``Subspace clustering," {\em IEEE Signal Processing Magazine}, 28(2):52--68, 2011

\bibitem{sujay1}
D. Park, C. Caramanis, and S. Sanghavi, ``Greedy Subspace Clustering,'' in Proc. {\em NIPS} 2014

\bibitem{aarti}
Y. Wang, Y. Wang, A. Singh, ``Clustering Consistent Sparse Subspace Clustering,''  {\em arXiv:1504.01046v1}, Apr. 2015

\bibitem{bal1}
B. Eriksson, L. Balzano, and R. Nowak,
``High-Rank Matrix Completion and Subspace Cluster-
ing with Missing Data'', in Proc. {\em Conference on Artificial Intelligence and Statistics (AI Stats)},
2012.

\bibitem{bal2}
L. Balzano, R. Nowak, A. Szlam, and B. Recht, ``k-Subspaces with missing data'', in Proc. {\em Statistical Signal Processing Workshop},
2012.
\bibitem{bal3}
D. Pimentel,  R. Nowak,  and L. Balzano, ``On the sample complexity of subspace clustering with missing data,'' in Proc. {\em IEEE Workshop on Statistical Signal Processing (SSP)} , vol., no., pp.280,283, June 29 2014-July 2 2014


\bibitem{bol2}
R. Heckel and H. Bolcskei, ``Robust subspace clustering via thresholding,''
{\em arXiv:1307.4891v2}, 2014.

\bibitem{wang}
Y.-X. Wang, H. Xu, and C. Leng, ``Provable subspace clustering: When LRR meets SSC," in
Proc. {\em Advances in Neural Information Processing Systems (NIPS)}, December 2013.

\bibitem{yih}Y. Wu and S. Verdu, ``Renyi information dimension: Fundamental limits of almost lossless analog compression,'' {\em IEEE Trans. Inf. Theory}, vol. 56,
no. 8, pp. 3721–3748, Aug. 2010.


\bibitem{bol}
Erwin Riegler, David Stotz, Helmut Bolcskei, ``Information-Theoretic Limits of Matrix Completion,'' {\em arXiv:1504.04970v2}, Apr 2015.

\bibitem{fra}
K. Falconer, {\em Fractal Geometry}, 1st ed.    New York, NY: Wiley, 1990.




\end{thebibliography}

\end{document}